\documentclass[conference]{IEEEtran}
\usepackage{amsmath,amssymb}
\usepackage{amsthm}
\usepackage{graphicx}
\usepackage{cite}
\usepackage{algorithm}
\usepackage{algpseudocode}
\usepackage{subfigure}
\usepackage{xcolor}
\usepackage[normalem]{ulem}
\usepackage{pgfplots}
\usepackage{pgfplotstable}
\pgfplotsset{compat=1.18}
\usepackage{hyperref}
\newtheorem{theorem}{Theorem}
\newtheorem{remark}{Remark}
\title{Learning-Accelerated Joint AP/RIS Optimization for Low-Latency Indoor mmWave XR Networks}

\author{
\IEEEauthorblockN{Ali Elghariani\IEEEauthorrefmark{1}, Mohammed Elbamby\IEEEauthorrefmark{4}, 
and Anis Elgabli\IEEEauthorrefmark{5}}
\IEEEauthorblockA{\IEEEauthorrefmark{1}Electrical Engineering Department, King Abdulaziz University, Saudi Arabia\\
\IEEEauthorrefmark{4}Telef\'onica Scientific Research, Barcelona, Spain\\
\IEEEauthorrefmark{5}Industrial and Systems Engineering Department\\ \IEEEauthorrefmark{5}Interdisciplinary Research Center for Communication Systems and Sensing (IRC-CSS),\\ King Fahd University of Petroleum \& Minerals (KFUPM), Saudi Arabia.
}
}

\begin{document}
\maketitle

\begin{abstract}

This paper studies a reconfigurable intelligent surface (RIS)-assisted indoor architecture for extended reality (XR) applications, where multiple access points (APs) simultaneously serve multiple users over a shared frequency channel under a centralized AP controller. We formulate a non-convex sum-rate maximization problem with coupled AP precoders and RIS phase shifts, and develop a centralized alternating optimization (AO) framework for joint design. We rigorously prove that the proposed AO guarantees monotonic improvement of the sum rate that converges to a coordinate-wise stationary point, despite the non-convexity of the resulting subproblems. Although AO provides reliable performance, its iterative nature limits real-time applicability as the system dimension grows. To overcome this limitation, we propose a supervised learning-assisted approach that leverages AO-generated solutions to significantly accelerate convergence.
Simulation results demonstrate that the proposed framework achieves near-optimal performance with substantially fewer AO iterations, which makes it well-suited for latency-critical RIS-assisted XR deployments.

\end{abstract}

\begin{IEEEkeywords}
RIS, 60GHz, mmWave, beamforming, alternating optimization, blockage, supervised learning.
\end{IEEEkeywords}

\section{Introduction}
Immersive extended reality (XR) applications such as collaborative gaming, virtual training, and interactive simulation are expected to place stringent demands on future indoor wireless networks\cite{ghourab2025extended}. These applications require not only multi-gigabit data rates, but also low and predictable latency to preserve real-time interaction and quality of experience. Millimeter-wave (mmWave) communication, particularly in the 60 GHz band, offers the large bandwidth needed to support such services\cite{redondi2024survey}. However, mmWave indoor links are highly vulnerable to blockage by walls, furniture, and human bodies, which can rapidly interrupt connectivity and degrade system performance.

A particularly challenging XR scenario arises when multiple access points (APs) must simultaneously serve multiple users over the same time-frequency resource in a dense indoor environment. In practice, some users may lose the direct path to their serving APs due to blockage, while still requiring uninterrupted high-rate service. This motivates the need for a network architecture that can restore connectivity through controllable reflected propagation while still supporting parallel multi-user transmission. Reconfigurable Intelligent Surfaces (RISs) provide a promising mechanism for this purpose by creating virtual line-of-sight paths through programmable reflection \cite{hassouna2023survey}. Early and recent RIS studies, such as \cite{iqbal2025comprehensive}, \cite{sun2024sum}, \cite{li2023mmwave} have investigated a variety of optimization techniques in single-cell scenarios. Authors in \cite{yang2024ris} studied the case of cooperative multi cell scenario in integrated sensing and communications applications. However, to the best of the authors' knowledge, the impact of inter-cell, or equivalently inter-AP, interference in multi-AP RIS-assisted systems remains insufficiently explored.

In this paper, we focus on a unique indoor network architecture that is well suited to XR environments. The architecture consists of multiple APs serving multiple users simultaneously on a single frequency channel, under the control of a centralized AP controller. The specific operating condition of interest is when multiple users experience blockage on their respective direct AP-to-user links, and communication must therefore be sustained through a single RIS deployed in the environment. In this setting, the RIS is not merely used to boost one link or extend the coverage of a single cell. Instead, one common RIS must simultaneously support several blocked users associated with different APs, while shaping the reflected signals so that useful transmission is enhanced and inter-AP interference is controlled.

This architecture leads to an important and largely underexplored design question: can a single RIS be configured to enable simultaneous multi-AP multi-user communication over one frequency while keeping interference manageable and ensuring reliable network performance? Addressing this question requires the RIS to act as a shared spatial coordination device across multiple AP-user pairs. Since the reflected paths generated by one AP may help one user while harming another, the RIS phase configuration must be jointly designed with the AP precoders so that the overall network sum-rate is improved rather than degraded. This joint design becomes even more important in the considered centralized architecture, where all APs are coordinated by one controller and can therefore be optimized together.

Motivated by this problem, we formulate a RIS-assisted sum-rate optimization problem for the above indoor architecture. The optimization jointly designs the transmit precoding of the APs under a zero-forcing (ZF) constraint and the phase shifts of the RIS elements under the assumption that some users' direct links are blocked and hence their communication takes place through the RIS. The problem is inherently non-convex due to the coupling between active beamforming and passive beamforming, and its size grows rapidly with the number of APs, users, and RIS elements. To solve it, we adopt an alternating optimization (AO) framework implemented at the centralized AP controller. In each iteration, the AP precoders and RIS phases are updated in turn, enabling coordinated interference-aware transmission through the shared RIS. Despite the non-convexity of the resulting subproblems, we show that the AO iterates converge to a coordinate-wise locally optimal solution of the discretized problem.

While centralized AO provides a structured way to solve the problem, its practical use in XR systems is challenged by computational latency. Because the optimization dimension becomes large in networks with many APs, users, and RIS elements, repeated AO iterations can be costly in online operation. This creates a gap between optimization quality and real-time deployability. To bridge this gap, we incorporate a learning-based acceleration strategy. Specifically, offline AO solutions are used to train a lightweight model that predicts a high-quality initialization for the joint design variables from the network state. During online operation, this learned initialization places the AO procedure much closer to the final solution, thereby reducing the number of required iterations and accelerating convergence.

\section{System Model}

Consider an indoor millimeter-wave communication system consisting of multi-antenna access points (APs), one AP controller (APC), multiple single-antenna user equipments (UEs), and one reconfigurable intelligent surface (RIS), as illustrated in Fig.~\ref{fig:systemModel}. The APs are equipped with $M$ antennas each, while the RIS consists of \(N\) passive elements arranged in an \(N_x \times N_y\) uniform planar array (UPA). The system operates at 60 GHz carrier frequency with wavelength \(\lambda = 5\) mm. A total number of APs \(J\) simultaneously serve \(K\) UEs, with \(J=K\). Note that in this XR system model, every AP serves a single user at a time, thus, there is no case such that $J \neq K$. The APs are positioned at fixed locations. A critical assumption in the system model is that the direct AP--UE links are completely blocked. Note that the assumptions $J=K$ and blocked direct links are adopted to simplify the analysis and isolate the shared-RIS interference coordination problem.

\begin{figure}[!t]
    \centering
    \includegraphics[width=0.38\textwidth,height=0.21\textwidth]{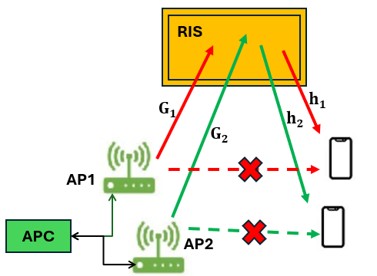}
    \caption{System model for two blocked users assisted by a single RIS.}
    \label{fig:systemModel}
\end{figure}

The wireless channel is modeled using the Saleh–Valenzuela (SV) model, in which propagation consists of a finite number of dominant clusters, each containing multiple rays with distinct complex gains, delays, and AoA/AoD parameters. This model is particularly suitable for mmWave systems as propagation at such high frequencies is typically sparse, and it captures both large-scale path loss and small-scale fading effects. Further details are in \cite{saleh1987statistical}, \cite{li2022ergodic}.

The AP\(_j\)-RIS channel \(\mathbf{G}_j \in \mathbb{C}^{N \times M}\) and RIS-UE\(_k\) channel \(\mathbf{h}_{k}^H \in \mathbb{C}^{1 \times N}\). Let \(\mathbf{\Theta} = \mathrm{diag}(\theta_1, \theta_2, \ldots, \theta_N)\) denote the RIS reflection matrix, where \(\theta_n = e^{j\phi_n}\) with \(\phi_n \in [0, 2\pi)\) representing the phase shift of the \(n\)-th RIS element. The effective cascaded channel from AP\(_j\) to UE\(_k\) via the RIS is given by
\[
h_{\text{eff},k,j} = \mathbf{h}_{k}^H \mathbf{\Theta} \mathbf{G}_j \mathbf{w}_j
\]
where \(\mathbf{w}_j \in \mathbb{C}^{M \times 1}\) is the precoding vector at AP\(_j\) with power constraint \(\|\mathbf{w}_j\|^2 \leq P_j\), and \(P_j = 1\) W is the maximum transmit power per AP.

The received signal at UE\(_k\) can be expressed as
\[
y_k = \underbrace{h_{\text{eff},k,j} s_j}_{\text{desired signal}} + \underbrace{\sum_{j \neq k} h_{\text{eff},k,j} s_j}_{\text{inter-user interference}} + n_k
\]
for \(k \in \{1,2,\ldots,K\}\), \(j \in \{1,2,\ldots,J\}\), and \(s_j \sim \mathcal{CN}(0,1)\) is the transmitted symbol from AP\(_j\) to UE\(_k\), and \(n_k \sim \mathcal{CN}(0,\sigma^2)\) is the additive white Gaussian noise.

In this system, we consider a centralized multi-AP architecture in which all APs are coordinated by a common AP controller (APC). The APC is assumed to have access to the channel state information (CSI) of all AP--RIS and RIS--UE links, which can be acquired through uplink pilot transmission and centralized CSI aggregation. Specifically, users transmit pilot sequences during the uplink training phase, each AP locally estimates the associated channels, and the resulting CSI is forwarded to the APC for centralized beamforming and RIS optimization.

\section{Problem Formulation}

For a given RIS configuration, assume that the initial values of $\boldsymbol{\Theta}$ can be obtained from random phase-shifts. The cascaded AP--RIS--UE channels between all APs and users are collected into the effective channel matrix
\begin{equation}
\mathbf{H}_{\text{eff}} =
\begin{bmatrix}
\mathbf{h}_1^{H} \boldsymbol{\Theta} \mathbf{G}_1 &
\cdots &
\mathbf{h}_1^{H} \boldsymbol{\Theta} \mathbf{G}_J \\
\vdots & \ddots & \vdots \\
\mathbf{h}_K^{H} \boldsymbol{\Theta} \mathbf{G}_1 &
\cdots &
\mathbf{h}_K^{H} \boldsymbol{\Theta} \mathbf{G}_J
\end{bmatrix}
\in \mathbb{C}^{K \times (J M)}.
\label{eq:Heff}
\end{equation}

Given $\mathbf{H}_{\mathrm{eff}}$, the APC computes the centralized zero-forcing (ZF) precoder as in (\ref{eq:Wzf}) below. To avoid numerical instability when the effective channel becomes ill-conditioned, the implementation employs the Moore Penrose pseudo inverse or regularized ZF. 
\begin{equation}
\mathbf{W}_{\mathrm{ZF}} =
\left(\mathbf{H}_{\mathrm{eff}}\right)^H
\left(\mathbf{H}_{\mathrm{eff}}\left(\mathbf{H}_{\mathrm{eff}}\right)^H\right)^{-1}.
\label{eq:Wzf}
\end{equation}
The transmit precoder associated with AP $j$ is then extracted and normalized to satisfy the per-AP power constraint as:
\begin{equation}
\mathbf{w}_j =
\sqrt{P_j}\,
\frac{[\mathbf{W}_{\mathrm{ZF}}]_{:,j}}
{\|[\mathbf{W}_{\mathrm{ZF}}]_{:,j}\|},
\quad \forall j \in \{1,\dots,J\}.
\label{eq:wj}
\end{equation}

Our objective is to jointly optimize the active AP precoders $\{\mathbf{w}_j\}_{j=1}^J$ and the passive RIS phase-shift matrix $\boldsymbol{\Theta}=\mathrm{diag}(e^{j\phi_1},\dots,e^{j\phi_N})$ in order to maximize the network sum rate. The resulting optimization problem is formulated as
\begin{equation}
\begin{aligned}
\max_{\{\mathbf{w}_j\},\,\boldsymbol{\Theta}}
\quad &
\sum_{k=1}^{K}
\log_2\!\left(
1+
\frac{
|\mathbf{h}_k^H \boldsymbol{\Theta}\mathbf{G}_k \mathbf{w}_k|^2
}{
\sum_{j\neq k}
|\mathbf{h}_k^H \boldsymbol{\Theta}\mathbf{G}_j \mathbf{w}_j|^2 + \sigma^2
}
\right) \\
\text{s.t.}\quad
& \|\mathbf{w}_j\|^2 \le P_j, \quad \forall j=1,\dots,J, \\
& |\Theta_{n,n}| = 1, \quad \forall n=1,\dots,N.
\end{aligned}
\label{eq:origOpt}
\end{equation}

Problem \eqref{eq:origOpt} is non-convex due to the coupled optimization variables and the unit-modulus constraints imposed on the RIS elements. In particular, the AP precoders and RIS phases are strongly coupled through the cascaded effective channels, which makes direct joint optimization intractable.

\subsection{Alternating Optimization Framework}
\label{subsec: AO}
To solve \eqref{eq:origOpt}, we adopt a centralized alternating optimization (AO) framework that iteratively updates the active AP precoders and passive RIS phase shifts. At each iteration, the algorithm alternates between: i) AP precoder optimization with fixed RIS phases, and ii) RIS phase optimization with fixed AP precoders.

For fixed RIS phases at iteration $t$, $\boldsymbol{\Theta}^{(t)}$, the AP precoders are updated using the ZF solution in \eqref{eq:Wzf}--\eqref{eq:wj}, which mitigates inter-user interference through centralized precoding. The updated precoders are accepted only if they yield a non-decreasing sum rate, i.e.,
\begin{equation}
R(\mathbf{W}^{(t+1)},\boldsymbol{\Theta}^{(t)})
\ge
R(\mathbf{W}^{(t)},\boldsymbol{\Theta}^{(t)}).
\label{eq:monotone_ZF}
\end{equation}

Next, for fixed precoders $\{\mathbf{w}_j^{(t+1)}\}$, the RIS phase shifts are optimized sequentially using discrete coordinate descent. Specifically, the $n$-th RIS element is updated over the discrete phase set
\begin{equation}
\mathcal{F}=
\left\{
0,\frac{2\pi}{L},\dots,\frac{2\pi(L-1)}{L}
\right\},
\label{eq:Fset}
\end{equation}
where $L$ denotes the RIS phase resolution. The phase of the $n$-th element is selected as
\begin{equation}
\phi_n^{(t+1)}=
\arg\max_{\phi\in\mathcal{F}}
R\!\left(
\mathbf{W}^{(t+1)},
\Theta\!\left(\phi \mid \{\phi_i^{(t+1)}\}_{i<n},\{\phi_i^{(t)}\}_{i>n}\right)
\right).
\label{eq:RIS_update}
\end{equation}

The AO iterations continue until convergence, which is declared when the relative sum-rate improvement satisfies
\begin{equation}
\frac{|R^{(t)}-R^{(t-1)}|}{R^{(t-1)}} < \epsilon,
\label{eq:convergence_condition}
\end{equation}
or when the maximum number of iterations $T_{\max}$ is reached.

The proposed AO framework guarantees monotonic convergence, since each subproblem is solved while fixing the other block of variables and each update yields a non-decreasing objective value.

\subsection{Convergence Analysis}

This subsection analyzes the convergence properties of the proposed alternating optimization algorithm, which serves as the core optimization engine in the overall framework. In the complete system, AO is used both as an offline solver to generate supervised learning labels and as an optional online refinement stage after the learned initialization. The convergence behavior of the proposed AO algorithm follows from the monotonicity of the ZF acceptance rule and the discrete coordinate-descent RIS update. The optimization problem is inherently non-convex. However, we prove that the proposed AO algorithm monotonically increases the system sum-rate and converges to a \emph{coordinatewise locally optimal point}.

\begin{theorem}[Finite Monotone Convergence and Coordinatewise Optimality of the AO Solution]
\label{lemma:finite_convergence}
For the discretized RIS phase set \(\mathcal{F}\), the proposed AO algorithm generates a non-decreasing sequence of sum-rate values and converges in a finite number of iterations to a coordinate-wise locally optimal solution of the discretized problem.
\end{theorem}

\begin{IEEEproof}
At iteration \(t\), the ZF update is accepted only if
\begin{equation}
R(\mathbf{W}^{(t+1)},\boldsymbol{\Theta}^{(t)})
\geq
R(\mathbf{W}^{(t)},\boldsymbol{\Theta}^{(t)}).
\end{equation}
For fixed \(\mathbf{W}^{(t+1)}\), each RIS phase element is updated by exact search over the finite set \(\mathcal{F}\). Hence, the RIS update also satisfies
\begin{equation}
R(\mathbf{W}^{(t+1)},\boldsymbol{\Theta}^{(t+1)})
\geq
R(\mathbf{W}^{(t+1)},\boldsymbol{\Theta}^{(t)}).
\end{equation}
Therefore,
\begin{equation}
R(\mathbf{W}^{(t+1)},\boldsymbol{\Theta}^{(t+1)})
\geq
R(\mathbf{W}^{(t)},\boldsymbol{\Theta}^{(t)}),
\end{equation}
which shows monotonic improvement. Each coordinate update strictly increases or maintains the objective, and the feasible RIS phase configurations belong to the finite discrete set \(\mathcal{F}^N\). Hence, the number of distinct achievable objective values is finite. Therefore, the AO process can only perform a finite number of strictly improving updates before reaching a configuration where no single coordinate update can further increase the objective. This point is coordinate-wise locally optimal
\end{IEEEproof}

\begin{remark}[Practical Stability]
Although ZF is a heuristic rather than an exact solution to the continuous precoder subproblem, the acceptance rule ensures monotonic improvement and empirical convergence stability.
\end{remark}

The proposed AO algorithm with ZF precoding and discrete coordinate descent RIS phase optimization generates a non-decreasing and bounded sequence of sum-rate values and converges in a finite number of steps to a coordinatewise locally optimal solution of the discretized problem.

\subsection{Learning-Assisted Online Approximation of AO}

Although the proposed AO algorithm provides a strong optimization framework and desirable convergence behavior, its iterative nature may lead to non-negligible latency in online operation. This issue is particularly important in XR applications, where channel conditions and user positions may change frequently and the RIS configuration must be updated with stringent delay requirements. To facilitate practical deployability, we propose a learning-assisted online approximation framework that leverages the AO algorithm offline and reduces the online computational burden.

The key idea is to use the AO solution as a source of supervision. Specifically, we generate a dataset offline by sampling many channel realizations and user/AP geometries, and then solving the joint beamforming and RIS optimization problem using the proposed AO algorithm until convergence. The resulting AO outputs serve as target labels for training a supervised learning model. The input features {are the real and imaginary parts of the AP--RIS channel matrices $\{\mathbf{G}_j\}$ and RIS--UE channel vectors $\{\mathbf{h}_k\}$, concatenated into a single vector}, while the output {consists of predicted discrete phase shift indices for each RIS element (via classification) and predicted AP beamforming vectors (via regression)}.

In the online phase, instead of starting AO from a random RIS configuration, the trained model directly predicts a high-quality solution estimate. This estimate can be used in two modes:
\begin{enumerate}
    \item \textbf{Direct inference mode:} the model output is used directly as the final configuration when very low latency is required.
    \item \textbf{Hybrid refinement mode:} the model output is used as a warm start, followed by a very small number of AO iterations to further improve the objective.
\end{enumerate}

This learning-assisted design is motivated by practical considerations. First, supervised learning is well suited here because the AO algorithm already provides reliable labels, avoiding the need for trial-and-error interaction with the environment. Second, compared with reinforcement learning, the proposed supervised framework is simpler to train, easier to interpret, and more straightforward to deploy. Third, by warm-starting AO from a learned solution rather than a random initialization, the online algorithm can achieve near-converged performance with drastically fewer iterations.

Let \(\mathcal{M}_{\psi}(\cdot)\) denote the trained model with parameters \(\psi\). For an online channel realization \(\mathbf{c}\), the model outputs
\begin{equation}
(\widehat{\mathbf{W}}^{(0)}, \widehat{\boldsymbol{\Theta}}^{(0)}) = \mathcal{M}_{\psi}(\mathbf{c})
\end{equation}
{where $\widehat{\boldsymbol{\Theta}}^{(0)}$ is decoded from predicted phase shift indices via $\widehat{\theta}_n = e^{j\frac{2\pi}{L}\widehat{\ell}_n}$.} {The decoded outputs are then either used directly or passed to AO as initialization.} The refined online solution after \(T_{\mathrm{ref}}\) additional AO iterations is expressed as
\begin{equation}
(\widehat{\mathbf{W}}, \widehat{\boldsymbol{\Theta}})
=
\mathrm{AO}\big(\widehat{\mathbf{W}}^{(0)}, \widehat{\boldsymbol{\Theta}}^{(0)}; T_{\mathrm{ref}}\big)
\end{equation}
where \(T_{\mathrm{ref}} \ll T_{\max}\) in practice.

The objective of the proposed learning-assisted framework is therefore not to replace AO entirely, but rather to approximate its converged solution with negligible online delay and, when needed, recover the remaining performance gap using only lightweight refinement iterations. 

\subsection{Complexity Analysis} 
Each AO iteration requires $\mathcal{O}((JM)^3)$ operations for ZF precoding and $\mathcal{O}(N L K)$ for RIS phase updates, resulting in an overall complexity of $\mathcal{O}(T_{\max}((JM)^3 + N L K))$. With the learning-assisted approach, the complexity reduces to $\mathcal{O}(T_{\mathrm{ref}}((JM)^3 + N L K))$ where $T_{\mathrm{ref}}<T_{\max} $. For large RIS sizes, the phase update stage dominates the runtime, as its complexity scales linearly with both the number of RIS elements and the phase-grid resolution, which motivates the proposed learning-assisted warm start.

\begin{algorithm}[t]
\small
\caption{Offline Training Procedure}
\label{alg:offline_training}
\begin{algorithmic}[1]
\Require Training size $N_{\rm tr}$, channel generator $\mathcal{G}(\cdot)$, AO solver, model $\mathcal{M}_{\psi}$
\Ensure Trained model $\mathcal{M}_{\psi^\star}$

\For{$n=1,\dots,N_{\rm tr}$}
    \State {Sample channels $\{\mathbf{G}_j^{(n)}\}, \{\mathbf{h}_k^{(n)}\} \sim \mathcal{G}(\cdot)$; form $\mathbf{c}^{(n)}$ from real/imag parts}
    \State Compute AO solution $(\mathbf{W}^{\star (n)}, \mathbf{\Theta}^{\star (n)})${; extract phase shift indices $\boldsymbol{\ell}^{\star (n)} \!\in\! \{0,\dots,L{-}1\}^N$}
    \State Store $\big(\mathbf{c}^{(n)},$ $\boldsymbol{\ell}^{\star (n)},$ $\mathbf{W}^{\star (n)})$
\EndFor
\State Train $\mathcal{M}_{\psi}$  {with cross-entropy loss on $\boldsymbol{\ell}$ and MSE loss on $\mathbf{W}$}
\State \Return $\mathcal{M}_{\psi^\star}$
\end{algorithmic}
\end{algorithm}

\begin{algorithm}[!t]
\small
\caption{Online Learning-Assisted AO Refinement}
\label{alg:online_learning_AO}
\begin{algorithmic}[1]
\Require Online channel realization $\mathbf{c}$, trained model $\mathcal{M}_{\psi^\star}$, refinement budget $T_{\rm ref}$
\Ensure Online solution $\widehat{\mathbf{W}}, \widehat{\mathbf{\Theta}}$

\State {Predict phase shift indices and beamformers:}
\[
{(\widehat{\boldsymbol{\ell}}^{(0)},\;} \widehat{\mathbf{W}}^{(0)}) = \mathcal{M}_{\psi^\star}(\mathbf{c})
\]
\State {Decode: $\widehat{\boldsymbol{\Theta}}^{(0)} = \mathrm{diag}\!\big(e^{j\frac{2\pi}{L}\widehat{\ell}_n^{(0)}}\big)_{n=1}^{N}$}
\If{$T_{\rm ref}=0$}
    \State $\widehat{\mathbf{W}} \gets \widehat{\mathbf{W}}^{(0)}$, \quad $\widehat{\mathbf{\Theta}} \gets \widehat{\mathbf{\Theta}}^{(0)}$
\Else
    \State Initialize AO with $(\widehat{\mathbf{W}}^{(0)}, \widehat{\mathbf{\Theta}}^{(0)})$
    \For{$t=1,\dots,T_{\rm ref}$}
        \State Iteratively update the active AP
precoders and passive RIS phase  as in (\ref{subsec: AO})
    \EndFor
    \State Set $(\widehat{\mathbf{W}}, \widehat{\mathbf{\Theta}})$ to the refined solution
\EndIf

\State \Return $\widehat{\mathbf{W}}, \widehat{\mathbf{\Theta}}$
\end{algorithmic}
\end{algorithm}

\section{Simulation and Results}

\subsection{{Learning Model Architecture and Training}}
\label{sec:ml_setup}

{We consider a multi-layer perceptron model with a shared encoder trunk and two task-specific output heads (\emph{SharedEncoderMLP}). The input vector $\mathbf{c} \in \mathbb{R}^{d_{\rm in}}$, with $d_{\rm in} = 2(JNM + KN)$, is formed by concatenating the real and imaginary parts of all AP--RIS channel matrices $\{\mathbf{G}_j\}$ and RIS--UE channel vectors $\{\mathbf{h}_k\}$.}

{The architecture has three stages: (i)~separate encoders for the $\mathbf{G}$ and $\mathbf{h}$ channel features using LayerNorm and BatchNorm respectively, (ii)~a shared two-layer trunk with BatchNorm, GELU activations, and dropout, and (iii)~two output heads. The $\boldsymbol{\theta}$-head produces $N \times L$ logits, considering each RIS element as an independent $L$-way classifier over the discrete phase set $\mathcal{F}$. The $\mathbf{W}$-head outputs $2JM$ real values encoding the AP beamformers. The composite training loss is}
\begin{equation}
{\mathcal{L} = \lambda_\theta \,\mathcal{L}_{\rm CE}(\boldsymbol{\ell}, \widehat{\boldsymbol{\ell}}) \;+\; \lambda_W \,\mathcal{L}_{\rm MSE}(\mathbf{W}, \widehat{\mathbf{W}})}
\end{equation}
{where $\mathcal{L}_{\rm CE}$ is cross-entropy on discrete phase shift classification and $\mathcal{L}_{\rm MSE}$ is MSE on beamformer regression. Training uses AdamW with cosine-annealing schedule and early stopping. Key parameters are listed in Table~\ref{tab:ml_params}. Hidden-layer widths scale with $N$ to accommodate different RIS sizes.}

\begin{table}[!t]
\centering
\caption{{Learning model training parameters.}}
\label{tab:ml_params}
{\begin{tabular}{l l}
\hline
\textbf{Parameter} & \textbf{Value} \\
\hline
Training / val.\ / test samples & 64\,000 / 1\,000 / 1\,000 \\
Batch size & 256 \\
Max epochs & 400 \\
Learning rate & $10^{-2}$ (cosine $\to 10^{-5}$) \\
Weight decay & $5 \times 10^{-4}$ \\
Loss weights $(\lambda_\theta, \lambda_W)$ & $(1.0,\; 0.5)$ \\
Phase resolution $L$ & 16 \\
Dropout & 0.10 \\
Early stopping patience & 60 epochs \\
\hline
\end{tabular}}
\end{table}

\subsection{{Numerical Results}}

In this section, numerical examples are provided to validate the proposed framework. We consider an indoor stage for XR experience of size \(5\)m \(\times 5\)m \(\times 3\)m with a single RIS positioned at (\(5, 2.5, 2\)) and with $N = 32, 64, 128$ RIS elements. We evaluate two network sizes of \(J \in \{2,4\}\) APs and \(K \in \{2,4\}\) users. Each AP employs \(M=8\) antennas and is mounted on the walls of the indoor environment, while each user has a single antenna and is uniformly distributed within the \(5\,\mathrm{m}\times5\,\mathrm{m}\) area, with height fixed at \(1.5\,\mathrm{m}\). The RIS is modeled with non-ideal reflection efficiency \(\eta=0.8\). The carrier frequency is \(60\,\mathrm{GHz}\) with transmission bandwidth \(2.1\,\mathrm{GHz}\) \cite{nitsche2014ieee}, and the noise power spectral density is set to \(-174\) dBm/Hz. As benchmarks, we compare the proposed scheme against i) TDM, where the RIS is optimized for one user at a time using the same AO procedure with equal time allocation, and ii) random RIS phase configuration without AO optimization. All channels follow the Saleh--Valenzuela model, and reported results are averaged over 1000 independent channel and user-location realizations.

For the learning-assisted framework, the dataset is generated offline using the converged AO solutions under many independent realizations. A supervised learning model is then trained to predict a high-quality initialization for the online phase. During testing, we evaluate three online modes: fully iterative AO from random initialization, direct prediction using the trained model, and hybrid prediction followed by a very small number of AO refinement iterations.

Fig.~\ref{fig:sumratevsbaseline} shows the achievable sum rate of the proposed algorithm at different \(J\) and \(K\) with RIS size equal to 128 elements. The figure emphasizes that at transmit power 30 dBm, which corresponds to a relatively high SNR regime, the joint optimization of precoder and RIS phases allows the system to exploit spatial diversity of the RIS in simultaneous communication. In both cases of \((J,K) = (2,2)\) and \((4,4)\), the joint optimization algorithm outperforms the TDM and random cases. Interestingly, as the number of APs and UEs increases from 2 to 4 in the same indoor stage size, the optimized RIS is still able to manage interference and achieve improved sum rate in the 4-AP case over the 2-AP case.

Fig.~\ref{fig:sumrate1vsN}, on the other hand, focuses on studying the effect of different RIS sizes of \(N = 32, 64,\) and \(128\) on the sum rate of the proposed algorithm when the transmit power is 30 dBm. It clearly demonstrates that the sum rate increases as the size of the RIS increases from 32 to 128 elements.

Fig.~\ref{fig:sumratevsSNR} illustrates the achievable sum rate of the joint optimization algorithm at various transmit powers when \(J = K = 4\). It can be seen that the achievable sum rate from all users increases as the transmit power increases, and this improvement becomes more pronounced when the transmit power exceeds 10 dBm.

The convergence behavior of the proposed AO algorithm was empirically evaluated, revealing a median convergence requirement of approximately 10--12 iterations at the 50th percentile when initialized randomly. While this confirms the stability of the iterative solver, it also highlights the practical limitation of AO for online deployment in latency-sensitive XR settings.

To address this issue, Fig.~\ref{fig:cdfLearning} shows the sum-rate CDF for $N \in \{64, 96, 128\}$ in three modes: AO-only, NN-only (direct model inference without AO refinement), and NN+AO (model prediction followed by AO refinements). The results show that NN-only approximates the AO-only CDF reasonably well, particularly at moderate $N$, which confirms that the learned model captures the essential structure of the AO solution. Moreover, with only a few refinement iterations, NN+AO closely approaches the performance of converged AO across all $N$ values and percentiles, demonstrating that the learned warm start provides near-AO performance with significantly lower iterative complexity.
We further quantify the iteration savings in Fig.~\ref{fig:iterate_fig}.  As $N$ grows from 16 to 128, AO requires considerably more iterations (from ${\sim}1.6$ to ${\sim}10.2$), whereas NN+AO converges with much fewer iterations (from ${\sim}0.1$ to ${\sim}4.7$), while both achieve nearly identical sum-rate. Due to space constraints, latency is reported for one representative case with $N=128$ in which the average online optimization time was 220 ms (AO-only), 0.5 ms (NN-only) and 98 ms (NN+AO), respectively. We have also tested the AO and the learning-assisted algorithms against imperfect CSI by imposing a normalized Gaussian channel estimation error of 0.1 and the sum-rate degradation remains below 1 bps/Hz. These results confirm that the proposed learning-assisted method substantially accelerates convergence, making it well suited for real-time RIS-assisted XR deployments.

\begin{figure}[!t]
    \centering
    \includegraphics[width=0.45\textwidth, height=0.2\textheight]{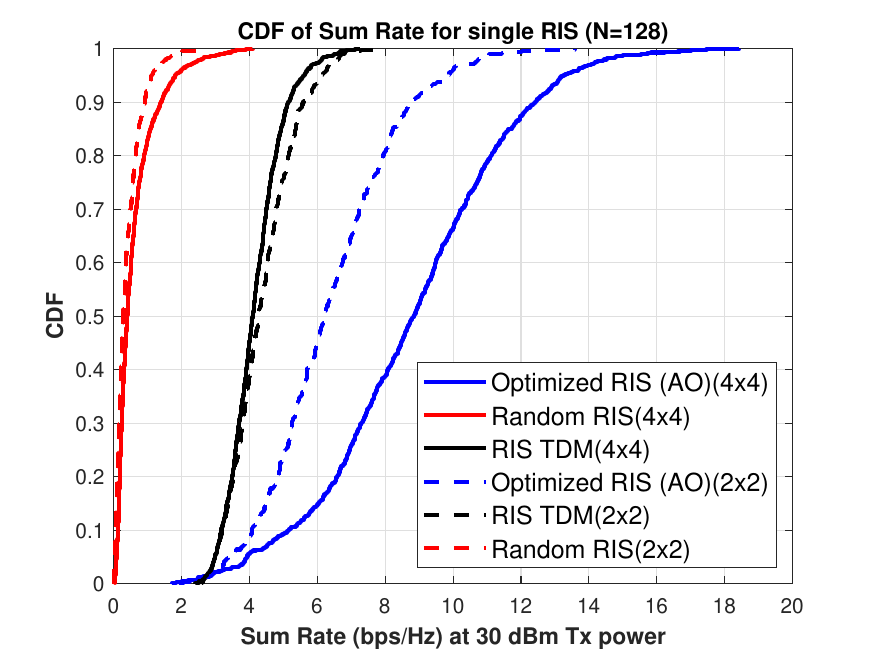}
     \vspace{-2mm}
    \caption{Sum-rate comparison at $N=128$ and $J=(2,4)$.}
    \label{fig:sumratevsbaseline}
\end{figure}

\begin{figure}[!t]
    \centering
    \includegraphics[width=0.45\textwidth, height=0.2\textheight]{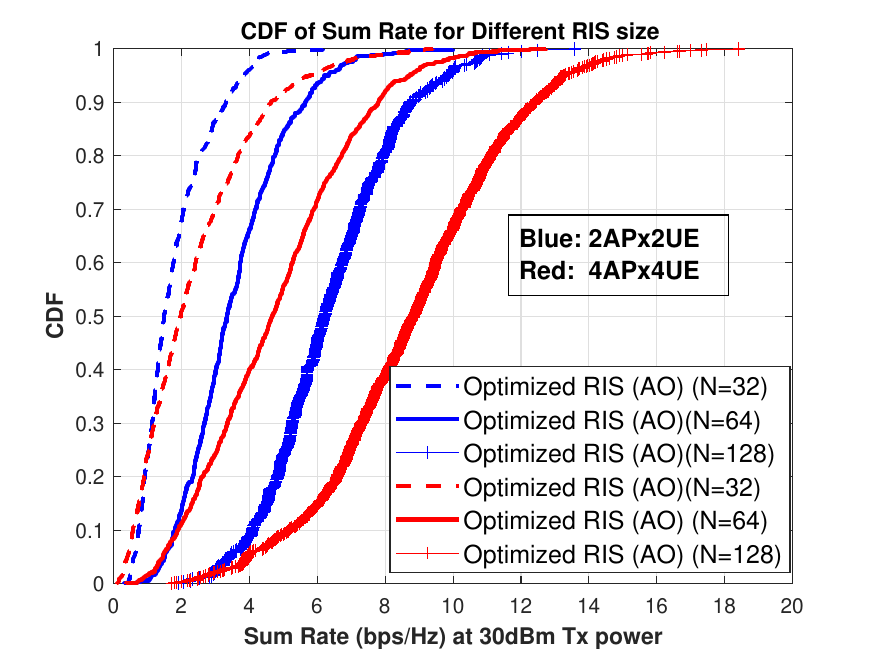}
     \vspace{-2mm}
    \caption{Optimized sum-rate at different $N$ and different $J,K$.}
    \label{fig:sumrate1vsN}
\end{figure}

\begin{figure}[!t]
    \centering
 \includegraphics[width=0.45\textwidth, height=0.205\textheight]{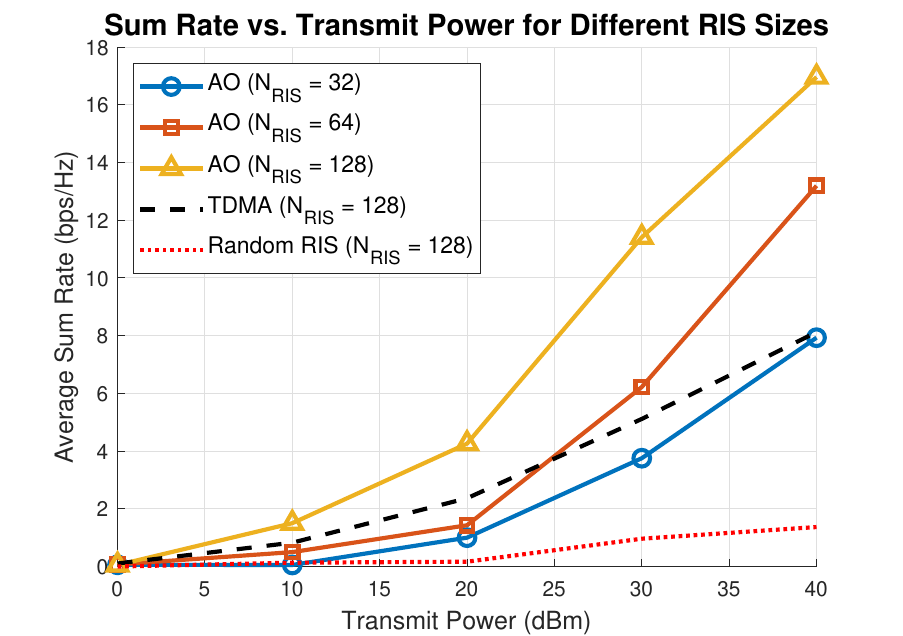}
  \vspace{-2mm}
    \caption{Sum-rate vs. transmit power with 4 APs \(\times\) 4 UEs at different $N$}
    \label{fig:sumratevsSNR}
\end{figure}

\begin{figure}[!t]
    \centering
    \includegraphics[width=0.4\textwidth, height=0.205\textheight]{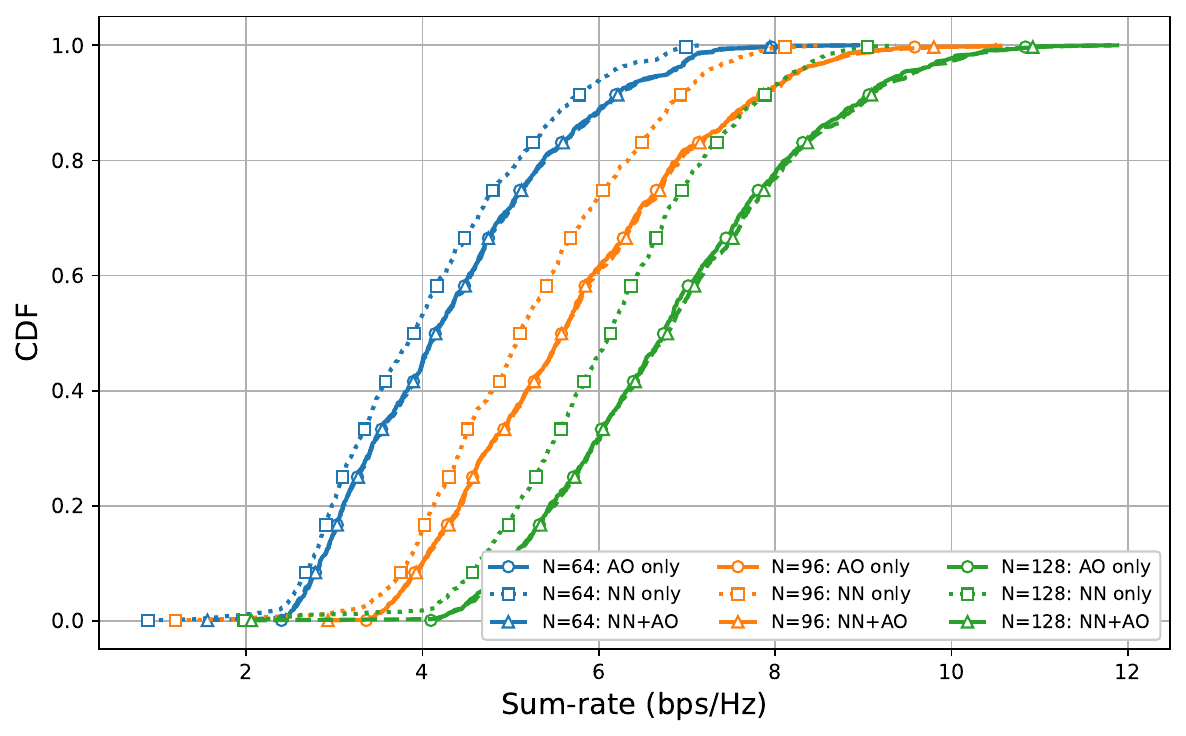}
     \vspace{-2mm}
    \caption{\small Sum-rate CDF for different $N$ values}
    \label{fig:cdfLearning}
\end{figure}

\begin{figure}[!t]
    \centering
    \includegraphics[width=0.4\textwidth,height=0.205\textheight]{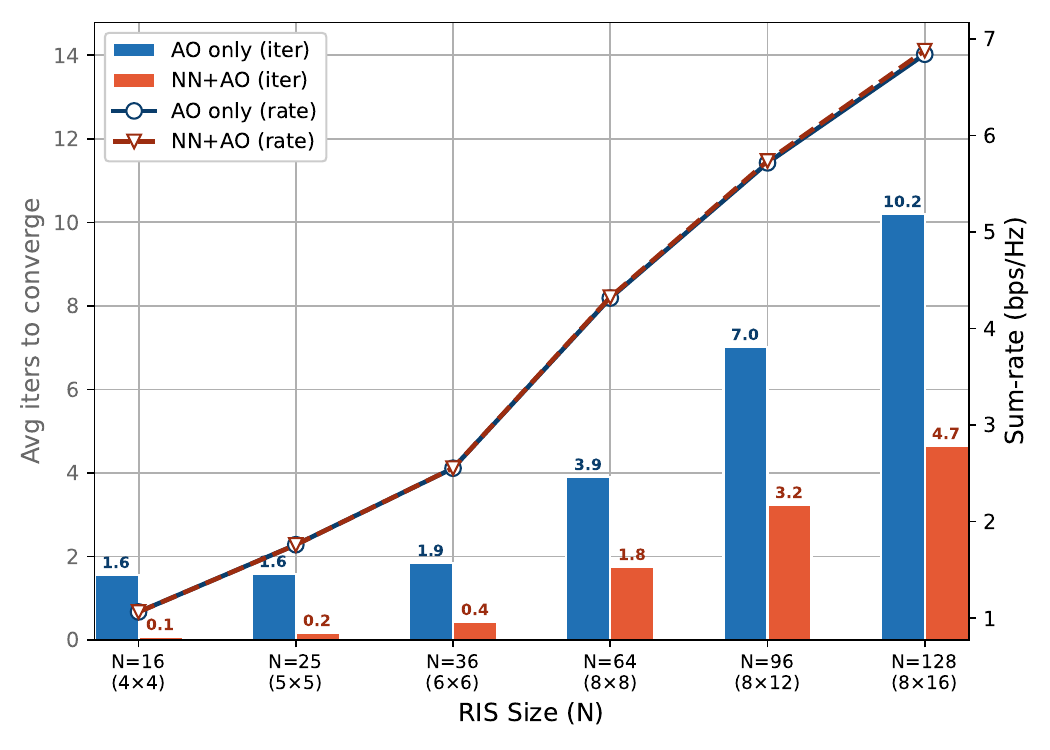}
     \vspace{-4mm}
    \caption{\small Iterations \& Sum-rate for different $N$ values.}
    \label{fig:iterate_fig}
\end{figure}

\section{Conclusion}

In this work, we formulated a sum-rate maximization problem for a shared single-RIS multi-AP  XR architecture and developed an AO-based solution for joint AP precoding and RIS phase design, which achieves monotonic convergence to a locally optimal solution. To reduce online complexity, we further proposed a supervised learning-assisted initialization strategy that significantly accelerates convergence while preserving near-optimal performance. The resulting hybrid model-based and learning-driven framework provides a practical solution for lowering online optimization time.

{\bf Acknowledgement}

This publication is based upon work supported by King Fahd University of 
Petroleum \& Minerals. Author(s) at KFUPM acknowledge the Deanship of Research at King Fahd University of
Petroleum \& Minerals for the support received under Grant no. IN26019. 

\bibliographystyle{IEEEtran}
\bibliography{referencesAO}

\end{document}